\newif\ifarxiv
\arxivtrue

\documentclass[aps,prl,twocolumn,superscriptaddress,nofootinbib]{revtex4-2} \usepackage{mypackage}

\begin{document}

\title{Composable simultaneous purification:\\
when all communication scenarios reduce to spatial correlations}

\date{\today}

\author{Matilde Baroni\ifarxiv $^\dagger$\footnotemark \fi }
    \affiliation{Sorbonne Université, CNRS, LIP6, 4 Place Jussieu, Paris F-75005, France}
\author{Dominik Leichtle\ifarxiv $^*$\footnotemark \fi }
    \affiliation{School of Informatics, University of Edinburgh, 10 Crichton Street, EH8 9AB Edinburgh, United Kingdom}
\author{Ivan \v{S}upi\'{c}}
\affiliation{Sorbonne Université, CNRS, LIP6, 4 Place Jussieu, Paris F-75005, France}
\affiliation{Université Grenoble Alpes, CNRS, Grenoble INP, LIG, 38000 Grenoble, France}
\author{Damian Markham}
    \affiliation{Sorbonne Université, CNRS, LIP6, 4 Place Jussieu, Paris F-75005, France}
\author{Marco Túlio Quintino}
    \affiliation{Sorbonne Université, CNRS, LIP6, 4 Place Jussieu, Paris F-75005, France}

\date{\today}

\begin{abstract}
    Bell non-locality is a powerful framework to distinguish classical, quantum and post-quantum resources, which relies on non-communicating players. Under which restriction can we have the same separations, if we allow for communication?
    Non-signalling state assemblages, and the fact that they can always be simultaneously purified, turned out to be the key element to restrict the simplest bipartite communication scenario, the prepare-and-measure, to the standard bipartite Bell scenario.
    Yet, many distinctive features of quantum theory are genuinely multipartite and cannot be reduced to two-party behaviour. In this work we are interested in extending this simultaneous purification inspired result to all multipartite communication schemes.
    As a first step, we unify and extend the simultaneous purification result from states to instruments and super-instruments, which are composable structures, and open up the possibility to explore more complex communication scenarios.
    Our main contribution is to establish that arbitrary compositions of non-signalling assemblages cannot escape the standard spatial quantum Bell correlations set.
    As a consequence, any interactive quantum realization of correlations outside of this set must involve at least one signalling assemblage of quantum operations, even when the resulting correlations are non-signalling.
    
\end{abstract}

\maketitle

\ifarxiv
    \begingroup
    \renewcommand\thefootnote{$^\dagger$} \footnotetext{\href{mailto:matilde.baroni@gmail.com}{matilde.baroni@gmail.com}}
    \renewcommand\thefootnote{$^*$} \footnotetext{\href{mailto:dominik.leichtle@ed.ac.uk}{dominik.leichtle@ed.ac.uk}}
    \endgroup
\fi

\noindent\textbf{Introduction.}
Purification theorems play a crucial role in quantum information, showing that mathematical models of quantum states and operations relate to physically realisable circuits.
Examples of this include Naimark's dilation theorem, stating that positive operator-valued measures (POVMs) can always be seen as compressions of projection-valued measures (PVMs), and Stinespring's dilation theorem,
which states that any quantum channel, \emph{i.e.}, completely positive trace-preserving (CPTP) linear map, can be written as a unitary operation that makes use of an auxiliary system which is later discarded.
Since it is always possible to purify a single state or operation, an ensemble of operations can always be trivially dilated by considering the direct sum of each dilation. However, this offers no insight into the relationship between the elements.
A long standing line of research has been concerned with the problem of finding appropriate conditions for \emph{simultaneous} purification, allowing assemblages of quantum states or operations to be realised by equivalent circuits, up to measurements on auxiliary registers.
The celebrated Schrödinger–Gisin–Hughston–Jozsa–Wootters (S-G-HJW) theorem~\cite{Schrödinger_1936,gisin1989stochastic,hughston1993complete} resolves this question for quantum states; see~\cite{kirkpatrick2005schrodingerhjwtheorem} for an overview.
A direct consequence of the S-G-HJW theorem is an equivalence between Bell non-local bipartite correlations and correlations arising in prepare-and-measure scenarios with preparation equivalences~\cite{Wright_2023}.

An analogous characterisation for correlations arising in multipartite communication scenarios involving more than two parties, or for correlations arising in bipartite scenarios with bidirectional communication, has so far been unknown.
This is the consequence of a lack of understanding of the problem of simultaneous purification in the more general case of assemblages of quantum operations.

In this work, we again shed light on the S-G-HJW theorem, originally formulated only for states, and present its natural extension to instrument and super-instrument assemblages.
Following Schrödinger's example~\cite{Schrödinger_1936}, we do not claim priority for these theorems, partially known in the C$^*$-algebraic framework as Radon-Nikodym theorems \cite{Raginsky, Chiribella_2013-RN-supermaps}, but the permission of exploring their implications in the following sections, for they are certainly not \emph{well} known.

In particular, we propose a definition of no-signalling assemblage that unifies the well-studied case of quantum states and more general cases of quantum operations, including quantum instruments and super-instruments, and demonstrate that it is the necessary and sufficient condition for simultaneous purification.
We proceed to show that this generalisation is well-behaved under sequential composition, composition along directed acyclic graphs, and composition with indefinite causal order.
The treatment of quantum super-instruments further allows us to study communication scenarios involving loops and players with internal memory.
Our main result is to show that multipartite quantum correlations arising in arbitrary communication scenarios, where each party implements no-signalling operations, can always be mapped back to the Bell non-local setting.

As a consequence, it follows that any interactive quantum realisation of correlations outside of the Bell quantum set must involve \emph{signalling} communication, even if the correlations are non-signalling.

\vspace{5pt}
\noindent\textbf{Non-signalling assemblages.}
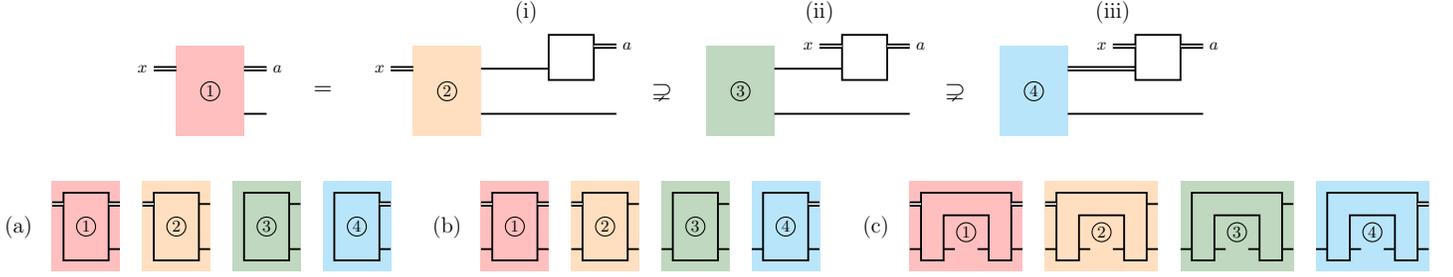
\begin{figure*}[t]
    \begin{center}
        \begin{tikzpicture}[scale=0.3, every node/.append style={scale=0.7}]

\begin{scope}[shift={(3,0)}]

\begin{scope}[shift={(-1.5,0)}]
    \fill[col1!25] (0,0) -- (3,0) -- (3,4) -- (0,4) -- cycle;
    \node[shape=circle, draw, inner sep=1pt] at (1.5,2) {1};
    \draw[double, line width=0.7pt] (0,3) -- (-1,3) node[left] {$x$};
    \draw[double, line width=0.7pt] (3,3) -- (4,3) node[right] {$a$};
    \draw[line width=0.7pt] (3,1) -- (4,1);
\end{scope}

\node at (14,5.5) {\large$(\text{i})$};
\node at (5,2) {\Large$=$};

\begin{scope}[shift={(9,0)}]
    \fill[col2!25] (0,0) -- (3,0) -- (3,4) -- (0,4) -- cycle;
    \node[shape=circle, draw, inner sep=1pt] at (1.5,2) {2};
    \draw[double, line width=0.7pt] (0,3) -- (-1,3) node[left] {$x$};
    
    \draw[line width=0.7pt] (3,3) -- (6,3);
    \draw[line width=0.7pt] (3,1) -- (9,1);

    \draw[line width=0.7pt] (6,2.5) -- (8,2.5) -- (8,4.5) -- (6,4.5) -- cycle;
    \draw[double, line width=0.7pt] (8,4) -- (9,4) node[right] {$a$};
\end{scope}

\node at (27,5.5) {\large$(\text{ii})$};
\node at (20,2) {\Large$\supsetneq$};

\begin{scope}[shift={(22,0)}]
    \fill[col3!25] (0,0) -- (3,0) -- (3,4) -- (0,4) -- cycle;
    \node[shape=circle, draw, inner sep=1pt] at (1.5,2) {3};
    
    \draw[line width=0.7pt] (3,3) -- (6,3);
    \draw[line width=0.7pt] (3,1) -- (9,1);

    \draw[line width=0.7pt] (6,2.5) -- (8,2.5) -- (8,4.5) -- (6,4.5) -- cycle;
    \draw[double, line width=0.7pt] (8,4) -- (9,4) node[right] {$a$};
    \draw[double, line width=0.7pt] (6,4) -- (5,4) node[left] {$x$};
\end{scope}

\node at (40,5.5) {\large$(\text{iii})$};
\node at (33,2) {\Large$\supsetneq$};

\begin{scope}[shift={(35,0)}]
   \fill[col4!25] (0,0) -- (3,0) -- (3,4) -- (0,4) -- cycle;
   \node[shape=circle, draw, inner sep=1pt] at (1.5,2) {4};
    
    \draw[double, line width=0.7pt] (3,3) -- (6,3);
    \draw[line width=0.7pt] (3,1) -- (9,1);

    \draw[line width=0.7pt] (6,2.5) -- (8,2.5) -- (8,4.5) -- (6,4.5) -- cycle;
    \draw[double, line width=0.7pt] (8,4) -- (9,4) node[right] {$a$};
    \draw[double, line width=0.7pt] (6,4) -- (5,4) node[left] {$x$};
\end{scope}

\end{scope} 

\begin{scope}[shift={(-4,-6)}]
    \node at (-1.5,2) {\large (a)};

    \fill[col1!25] (0,0) -- (3,0) -- (3,4) -- (0,4) -- cycle;
    \fill[col2!25] (4,0) -- (7,0) -- (7,4) -- (4,4) -- cycle;
    \fill[col3!25] (8,0) -- (11,0) -- (11,4) -- (8,4) -- cycle;
    \fill[col4!25] (12,0) -- (15,0) -- (15,4) -- (12,4) -- cycle;

    \node[shape=circle, draw, inner sep=1pt] at (1.5,2) {1};
    \node[shape=circle, draw, inner sep=1pt] at (5.5,2) {2};
    \node[shape=circle, draw, inner sep=1pt] at (9.5,2) {3};
    \node[shape=circle, draw, inner sep=1pt] at (13.5,2) {4};

    \draw[line width=0.7pt](0.5,0.5) -- (2.5,0.5) -- (2.5,3.5) -- (0.5,3.5) -- cycle;
    \draw[double, line width=0.7pt] (0,3) -- (0.5,3);
    \draw[double, line width=0.7pt] (2.5,3) -- (3,3);
    \draw[line width=0.7pt] (2.5,1) -- (3,1);

    \draw[line width=0.7pt] (4.5,0.5) -- (6.5,0.5) -- (6.5,3.5) -- (4.5,3.5) -- cycle;
    \draw[double, line width=0.7pt] (4,3) -- (4.5,3);
    \draw[line width=0.7pt] (6.5,3) -- (7,3);
    \draw[line width=0.7pt] (6.5,1) -- (7,1);

    \draw[line width=0.7pt] (8.5,0.5) -- (10.5,0.5) -- (10.5,3.5) -- (8.5,3.5) -- cycle;
    \draw[line width=0.7pt] (10.5,3) -- (11,3);
    \draw[line width=0.7pt] (10.5,1) -- (11,1);

    \draw[line width=0.7pt] (12.5,0.5) -- (14.5,0.5) -- (14.5,3.5) -- (12.5,3.5) -- cycle;
    \draw[double, line width=0.7pt] (14.5,3) -- (15,3);
    \draw[line width=0.7pt] (14.5,1) -- (15,1);
\end{scope}

\begin{scope}[shift={(15,-6)}]
    \node at (-1.5,2) {\large (b)};
    
    \fill[col1!25] (0,0) -- (3,0) -- (3,4) -- (0,4) -- cycle;
    \fill[col2!25] (4,0) -- (7,0) -- (7,4) -- (4,4) -- cycle;
    \fill[col3!25] (8,0) -- (11,0) -- (11,4) -- (8,4) -- cycle;
    \fill[col4!25] (12,0) -- (15,0) -- (15,4) -- (12,4) -- cycle;

    \node[shape=circle, draw, inner sep=1pt] at (1.5,2) {1};
    \node[shape=circle, draw, inner sep=1pt] at (5.5,2) {2};
    \node[shape=circle, draw, inner sep=1pt] at (9.5,2) {3};
    \node[shape=circle, draw, inner sep=1pt] at (13.5,2) {4};

    \draw[line width=0.7pt](0.5,0.5) -- (2.5,0.5) -- (2.5,3.5) -- (0.5,3.5) -- cycle;
    \draw[double, line width=0.7pt] (0,3) -- (0.5,3);
    \draw[double, line width=0.7pt] (2.5,3) -- (3,3);
    \draw[line width=0.7pt] (2.5,1) -- (3,1);
    \draw[line width=0.7pt] (0,1) -- (0.5,1);

    \draw[line width=0.7pt] (4.5,0.5) -- (6.5,0.5) -- (6.5,3.5) -- (4.5,3.5) -- cycle;
    \draw[double, line width=0.7pt] (4,3) -- (4.5,3);
    \draw[line width=0.7pt] (6.5,3) -- (7,3);
    \draw[line width=0.7pt] (6.5,1) -- (7,1);
    \draw[line width=0.7pt] (4,1) -- (4.5,1);

    \draw[line width=0.7pt] (8.5,0.5) -- (10.5,0.5) -- (10.5,3.5) -- (8.5,3.5) -- cycle;
    \draw[line width=0.7pt] (10.5,3) -- (11,3);
    \draw[line width=0.7pt] (10.5,1) -- (11,1);
    \draw[line width=0.7pt] (8,1) -- (8.5,1);

    \draw[line width=0.7pt] (12.5,0.5) -- (14.5,0.5) -- (14.5,3.5) -- (12.5,3.5) -- cycle;
    \draw[double, line width=0.7pt] (14.5,3) -- (15,3);
    \draw[line width=0.7pt] (14.5,1) -- (15,1);
    \draw[line width=0.7pt] (12,1) -- (12.5,1);
\end{scope}

\begin{scope}[shift={(34,-6)}]
    \node at (-1.5,2) {\large (c)};
    
    \fill[col1!25] (0,0) -- (5,0) -- (5,4) -- (0,4) -- cycle;
    \fill[col2!25] (6,0) -- (11,0) -- (11,4) -- (6,4) -- cycle;
    \fill[col3!25] (12,0) -- (17,0) -- (17,4) -- (12,4) -- cycle;
    \fill[col4!25] (18,0) -- (23,0) -- (23,4) -- (18,4) -- cycle;

    \node[shape=circle, draw, inner sep=1pt] at (2.5,1.75) {1};
    \node[shape=circle, draw, inner sep=1pt] at (8.5,1.75) {2};
    \node[shape=circle, draw, inner sep=1pt] at (14.5,1.75) {3};
    \node[shape=circle, draw, inner sep=1pt] at (20.5,1.75) {4};

    \draw[line width=0.7pt] (0.5,0.5) -- (1.5,0.5) -- (1.5,2.5) -- (3.5,2.5) -- (3.5,0.5) -- (4.5,0.5) --(4.5,3.5) -- (0.5,3.5) -- cycle;
    \draw[double, line width=0.7pt] (0,3) -- (0.5,3);
    \draw[double, line width=0.7pt] (4.5,3) -- (5,3);
    \draw[line width=0.7pt] (0,1) -- (0.5,1);
    \draw[line width=0.7pt] (1.5,1) -- (2,1);
    \draw[line width=0.7pt] (3,1) -- (3.5,1);
    \draw[line width=0.7pt] (4.5,1) -- (5,1);

    \draw[line width=0.7pt] (6.5,0.5) -- (7.5,0.5) -- (7.5,2.5) -- (9.5,2.5) -- (9.5,0.5) -- (10.5,0.5) -- (10.5,3.5) -- (6.5,3.5) -- cycle;
    \draw[double, line width=0.7pt] (6,3) -- (6.5,3);
    \draw[line width=0.7pt] (10.5,3) -- (11,3);
    \draw[line width=0.7pt] (6,1) -- (6.5,1);
    \draw[line width=0.7pt] (7.5,1) -- (8,1);
    \draw[line width=0.7pt] (9,1) -- (9.5,1);
    \draw[line width=0.7pt] (10.5,1) -- (11,1);

    \draw[line width=0.7pt] (12.5,0.5) -- (13.5,0.5) -- (13.5,2.5) -- (15.5,2.5) -- (15.5,0.5) -- (16.5,0.5) -- (16.5,3.5) -- (12.5,3.5) -- cycle;
    \draw[line width=0.7pt] (16.5,3) -- (17,3);
    \draw[line width=0.7pt] (12,1) -- (12.5,1);
    \draw[line width=0.7pt] (13.5,1) -- (14,1);
    \draw[line width=0.7pt] (15,1) -- (15.5,1);
    \draw[line width=0.7pt] (16.5,1) -- (17,1);

    \draw[line width=0.7pt] (18.5,0.5) -- (19.5,0.5) -- (19.5,2.5) -- (21.5,2.5) -- (21.5,0.5) -- (22.5,0.5) -- (22.5,3.5) -- (18.5,3.5) -- cycle;
    \draw[double, line width=0.7pt] (22.5,3) -- (23,3);
    \draw[line width=0.7pt] (18,1) -- (18.5,1);
    \draw[line width=0.7pt] (19.5,1) -- (20,1);
    \draw[line width=0.7pt] (21,1) -- (21.5,1);
    \draw[line width=0.7pt] (22.5,1) -- (23,1);
\end{scope}
  
\end{tikzpicture}     \end{center}
    \caption{\justifying Decompositions for (a) state, (b) instrument, and (c) super-instrument assemblages.
        The first equality (i) is true for any valid assemblage.
        Decomposition~(ii) is valid if and only if the assemblage is non-signalling.
        The last decomposition (iii) refers to unsteerable assemblages.
        These separations reflect the separation between non-signalling, quantum, and classical sets of correlations.
        The decompositions can be generalized beyond super-instruments to arbitrary quantum objects, as shown in Appendix~\ref{app:SGHJW-proofs}.}
    \label{fig:decompositions}
\end{figure*}
In quantum mechanics, measurements are probabilistic processes, described mathematically by POVMs, i.e. a collection of positive bounded operators $M_a \in \mathsf{B}(\mathcal{H})$ labelled by a classical output $a \in \mathcal{A}$, that sum to the identity.
Often one wants to account for multiple measurement setting; this naturally leads to sets of operators $M_{a|x}$, labelled by an additional classical input $x \in \mathcal{X}$, such that each subset with a fixed $\bar{x}$, $\{M_{a|\bar{x}}\}_a$,  is a valid POVM.

It is natural to allow for the same structure for states.
A \textit{state assemblage} is a collection of sub-normalized quantum states $\{\rho_{a|x}\}_{a,x}$ on $\mathcal{H}$, such that for any fixed $\bar{x}$ the average over the outcome $a$, $\sum_{a}\rho_{a|\bar{x}} = \rho_{\bar{x}}$, is a normalised state.
We can always think of these as the collection of all post-measurement states produced by a measurement set $M_{a|\bar{x}}$ on a fixed state.

Consider the scenario where this fixed state is a shared bipartite resource $\rho_{AB}$. The first party, Alice, aims to remotely prepares the state assemblage $\rho_{a|x}$ for the second party, Bob, by performing a measurement $M_{a|x}$ on her system $A$. To generate all state assemblages allowed by quantum mechanics, the shared resource would generally have to depend on the input setting $x$ (Fig.~\ref{fig:decompositions}, decomposition~(i)).

What if only Alice has access to the variable $x$? Which assemblages can she remotely prepare for Bob?
The celebrated S-G-HJW theorem states that she can remotely prepare all assemblages that satisfy the \textit{no-signalling} property:
    \begin{equation}
        \sum_a \rho_{a|x} = \rho_x = \rho \qquad \forall x.
    \end{equation}
This equation encodes the causal requirement that the marginal state on Bob's side cannot depend on Alice's input (decomposition (ii) in Fig.~\ref{fig:decompositions}).

To be precise, the S-G-HJW theorem states that an assemblage $\{\rho_{a|x}\}_{a,x}$ on a finite-dimensional Hilbert space $\mathcal{H}$ is non-signalling if and only if 
there exist an auxiliary Hilbert space $\mathcal{H}_A$, a set of POVMs on this space $\{M_{a|x}\}_{a,x} \subseteq \mathsf{B}(\mathcal{H}_A)$ and a pure bipartite state $\ket{\psi} \in \mathcal{H}_A \otimes \mathcal{H}$ such that the assemblage can be remotely prepared:
\begin{equation}\label{eq:theorem-states}
    \sum_a \rho_{a|x} = \rho \hspace{3pt} \forall x \hspace{3pt} \iff \hspace{3pt} \rho_{a|x}= \Tr_A\left[ \left(M_{a|x} \otimes \mathds{1}\right) \ketbra{\psi}{\psi} \right].
\end{equation}
This statement can be equivalently represented graphically:
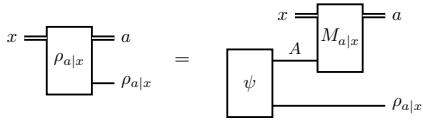
\begin{center}
    \begin{tikzpicture}[scale=0.3, every node/.append style={scale=0.7}]

\begin{scope}[shift={(0,0)}]
    \draw[line width=0.7pt](0,0) -- (2,0) -- (2,3) -- (0,3) -- cycle;
    \node at (1,1.5) {$\rho_{a|x}$};
    \draw[double, line width=0.7pt] (0,2.5) -- (-1,2.5) node[left] {$x$};
    \draw[double, line width=0.7pt] (2,2.5) -- (3,2.5) node[right] {$a$};
    \draw[line width=0.7pt] (2,0.5) -- (3,0.5) node[right] {$\rho_{a|x}$};
\end{scope}

\node at (6,1.5) {\large$=$};

\begin{scope}[shift={(8,-1)}]
    \draw[line width=0.7pt](0,0) -- (2,0) -- (2,3) -- (0,3) -- cycle;
    \node at (1,1.5) {$\psi$};
    
    \draw[line width=0.7pt] (2,2.5) -- (4,2.5);
    \node[above] at (3,2.5) {\small $A$};
    \draw[line width=0.7pt] (2,0.5) -- (7,0.5) node[right] {$\rho_{a|x}$};

    \draw[line width=0.7pt] (4,2) -- (6,2) -- (6,5) -- (4,5) -- cycle;
    \node at (5,3.5) {$M_{a|x}$};
    \draw[double, line width=0.7pt] (4,4.5) -- (3,4.5) node[left] {$x$};
    \draw[double, line width=0.7pt] (6,4.5) -- (7,4.5) node[right] {$a$};
\end{scope}

\end{tikzpicture} \end{center}
where single lines represent quantum states, double lines are classical information, and time flows from left to right.
For completeness, we report a constructive proof of this result in Appendix~\ref{app:SGHJW-proofs}.

An analogous statement holds for (probabilistic) transformations of quantum states;
this was originally proven by~\cite{BS86radon} in the C*-algebraic framework, reformulated by~\cite{Raginsky} in the language of quantum information, and later considered by~\cite{Uola_2018} in the context of non-signalling assemblages of transformations.
Such operations can be mathematically described by quantum instruments $\mathcal{I}_a$,  \emph{i.e.}, a collection of completely positive (CP) maps whose sum is also trace-preserving (TP).
As before, we are interested in sets of multiple instruments labelled by a classical input $x \in \mathcal{X}$.
An \textit{instrument assemblage} is a collection of maps $\{\mathcal{I}_{a|x}\}_{a,x} \subseteq \mathsf{CP}(\mathcal{H}_I, \mathcal{H}_O) \subseteq \operatorname{Lin}(\mathsf{B}(\mathcal{H}_I),\mathsf{B}(\mathcal{H}_O))$ such that, for each $x$, the sum
$\sum_{a}\mathcal{I}_{a|x} = \mathcal{I}_x \in \mathsf{CPTP}(\mathcal{H}_I, \mathcal{H}_O)$ is a channel.
If this channel does not depend on $x$, the assemblage is said to be 
non-signalling:
\vspace{-5pt}
\begin{equation}
    \sum_a \mathcal{I}_{a|x} = \mathcal{I}_x= \mathcal{I} \qquad \forall x.
\end{equation}
An equivalent of the S-G-HJW theorem can be formulated for instruments.
An instrument assemblage $\{\mathcal{I}_{a|x}\}_{a,x} \subseteq \mathsf{CP}(\mathcal{H}_I, \mathcal{H}_O)$ is 
non-signalling if and only if it admits a decomposition into a classical-input-independent isometric channel $\Lambda (\cdot) = V (\cdot) V^\dagger \in \mathsf{CPTP}(\mathcal{H}_I, \mathcal{H}_A \otimes \mathcal{H}_O)$, where $\mathcal{H}_A$ is an auxiliary Hilbert space and $V$ an isometry, and a classical-input-dependent measurement in this auxiliary space $\{M_{a|x}\}_{a,x} \subseteq \mathsf{B}(\mathcal{H}_A)$.
In equations, this reads as:
\begin{align}
         \sum_a \mathcal{I}_{a|x} = \mathcal{I} \iff  \mathcal{I}_{a|x}(\sigma) = \Tr_A\left[\left(M_{a|x}\otimes \mathds{1}\right)\Lambda(\sigma) \right],
\end{align}
or equivalently, graphically as follows:
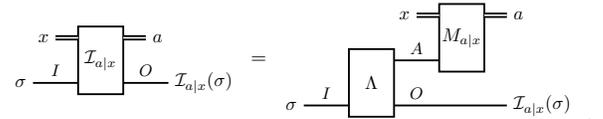
\begin{center}
    \begin{tikzpicture}[scale=0.3, every node/.append style={scale=0.7}]

\begin{scope}[shift={(0,0)}]
    \draw[line width=0.7pt](0,0) -- (2,0) -- (2,3) -- (0,3) -- cycle;
    \node at (1,1.5) {$\mathcal{I}_{a|x}$};
    \draw[double, line width=0.7pt] (0,2.5) -- (-1,2.5) node[left] {$x$};
    \draw[double, line width=0.7pt] (2,2.5) -- (3,2.5) node[right] {$a$};
    \draw[line width=0.7pt] (0,0.5) -- (-2,0.5) node[left] {$\sigma$};
    \node[above] at (-1,0.5) {\small $I$};
    \draw[line width=0.7pt] (2,0.5) -- (4,0.5) node[right] {$\mathcal{I}_{a|x}(\sigma)$};
    \node[above] at (3,0.5) {\small $O$};
\end{scope}

\node at (8,1.5) {\large$=$};

\begin{scope}[shift={(12,-1)}]
    \draw[line width=0.7pt](0,0) -- (2,0) -- (2,3) -- (0,3) -- cycle;
    \node at (1,1.5) {$\Lambda$};
    \draw[line width=0.7pt] (0,0.5) -- (-2,0.5) node[left] {$\sigma$};
    \node[above] at (-1,0.5) {\small $I$};
    
    \draw[line width=0.7pt] (2,2.5) -- (4,2.5);
    \node[above] at (3,2.5) {\small $A$};
    \draw[line width=0.7pt] (2,0.5) -- (7,0.5) node[right] {$\mathcal{I}_{a|x}(\sigma)$};
    \node[above] at (3,0.5) {\small $O$};

    \draw[line width=0.7pt] (4,2) -- (6,2) -- (6,5) -- (4,5) -- cycle;
    \node at (5,3.5) {$M_{a|x}$};
    \draw[double, line width=0.7pt] (4,4.5) -- (3,4.5) node[left] {$x$};
    \draw[double, line width=0.7pt] (6,4.5) -- (7,4.5) node[right] {$a$};
\end{scope}

\end{tikzpicture} .
\end{center}
This can be shown explicitly using the Choi–Jamiołkowski isomorphism, which provides a standard method for mapping instruments to states; for more details, consult Appendix \ref{app:SGHJW-proofs}.

Let us finally consider assemblages of super-maps and super-instruments, which are probabilistic transformations between quantum channels~\cite{Chiribella_2008,Chiribella_2009-superinstruments,taranto2025higherorderquantumoperations}.
Formally, a \textit{super-instrument assemblage} is a collection of super-maps $\{\mathcal{S}_{a|x}\}_{a,x}$, where each $\mathcal{S}_{a|x} : \operatorname{Lin}(\mathsf{B}(\mathcal{H}_{I_1}),\mathsf{B}(\mathcal{H}_{O_1})) \to \operatorname{Lin}(\mathsf{B}(\mathcal{H}_{I_2}),\mathsf{B}(\mathcal{H}_{O_2}))$ is completely CP preserving~\cite{Quintino_2019},
and such that $\sum_a \mathcal{S}_{a|x} = \mathcal{S}_x$ is a valid super-channel for every choice of $x$, \emph{i.e.}, $\mathcal{S}_x$ is TP preserving.
If the super-channels $\mathcal{S}_x $ do not depend on $x$, the assemblage is called non-signalling:
\begin{equation}
    \sum_a \mathcal{S}_{a|x} = \mathcal{S}_x=\mathcal{S} \qquad \forall x.
\end{equation}
A set of quantum super-instruments $\{\mathcal{S}_{a|x}\}_{a,x}$  on finite-dimensional Hilbert spaces is non-signalling if and only if it has a bipartite quantum realisation \cite{Chiribella_2013-RN-supermaps}, \emph{i.e.}, there exist an auxiliary Hilbert space $\mathcal{H}_A$, a POVM on this space $\{M_{a|x}\}_{a,x} \subseteq \mathsf{B}(\mathcal{H}_A)$, and a pure\footnote{We call super-channels $T : \operatorname{Lin}(\mathsf{B}(\mathcal{H}_{I_1}),\mathsf{B}(\mathcal{H}_{O_1})) \to \operatorname{Lin}(\mathsf{B}(\mathcal{H}_{I_2}),\mathsf{B}(\mathcal{H}_{O_2}))$ pure if they can be written in the form
$T(\mathcal{E})(\sigma) = W (\mathcal{E} \otimes \mathds{1}_E)(V \sigma V^\dagger) W^\dagger$
for some (finite-dimensional) auxiliary Hilbert space $\mathcal{H}_E$, and isometries
$V : \mathcal{H}_{I_2} \to \mathcal{H}_{I_1} \otimes \mathcal{H}_{E}$ and 
$W : \mathcal{H}_{O_1} \otimes \mathcal{H}_E \to \mathcal{H}_{O_2}$~\cite{Chiribella_2008,Chiribella_2009-superinstruments,araujo2017purification}.
}
super-channel $\Sigma : \operatorname{Lin}(\mathsf{B}(\mathcal{H}_{I_1}),\mathsf{B}(\mathcal{H}_{O_1})) \to \operatorname{Lin}(\mathsf{B}(\mathcal{H}_{I_2}),\mathsf{B}(\mathcal{H}_{O_2} \otimes \mathcal{H}_A))$ such that
\begin{align}
         \sum_a \mathcal{S}_{a|x} = \mathcal{S} \iff \mathcal{S}_{a|x}(\mathcal{E}) = \Tr_A\left[\left(M_{a|x}\otimes \mathds{1}\right)\Sigma(\mathcal{E}) \right],
    \end{align}
or equivalently graphically,
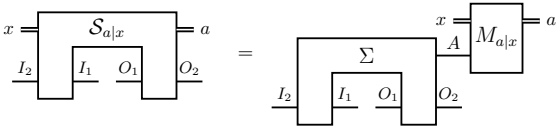
\begin{center}
    \begin{tikzpicture}[scale=0.23, every node/.append style={scale=0.7}]
\begin{scope}[shift={(-2,0)}]
\draw[line width=0.7pt] (1,0) -- (3,0) -- (3,3) -- (7,3) -- (7,0) -- (9,0) --(9,5) -- (1,5) -- cycle;

    \draw[double, line width=0.7pt] (1,4) -- (0,4) node[left] {$x$};
    \draw[double, line width=0.7pt] (9,4) -- (10,4) node[right] {$a$};

    \draw[line width=0.7pt] (-0.5,1) -- (1,1) node[midway, above] {\small $I_2$};
    \draw[line width=0.7pt] (3,1) -- (4.5,1) node[midway, above] {\small $I_1$};
    \draw[line width=0.7pt] (5.5,1) -- (7,1) node[midway, above] {\small $O_1$};
    \draw[line width=0.7pt] (9,1) -- (10.5,1) node[midway, above] {\small $O_2$};

    \node at (5,4) {\large$\mathcal{S}_{a|x}$};
    \end{scope} 

    \node at (11,2.5) {\large$=$};

\draw[line width=0.7pt] (14,-1.5) -- (16,-1.5) -- (16,1.5) -- (20,1.5) -- (20,-1.5) -- (22,-1.5) -- (22,3.5) -- (14,3.5) -- cycle;

\draw[line width=0.7pt] (12.5,-0.5) -- (14,-0.5) node[midway, above] {\small $I_2$};
\draw[line width=0.7pt] (16,-0.5) -- (17.5,-0.5) node[midway, above] {\small $I_1$};
\draw[line width=0.7pt] (18.5,-0.5) -- (20,-0.5) node[midway, above] {\small $O_1$};
\draw[line width=0.7pt] (22,-0.5) -- (23.5,-0.5) node[midway, above] {\small $O_2$};

\node at (18,2.5) {\large$\Sigma$};

\draw[line width=0.7pt] (22,2.5) -- (24,2.5) node[midway, above] {$A$};

\draw[line width=0.7pt] (24,1.5) -- (27,1.5) -- (27,5.5) -- (24,5.5) -- cycle;

\draw[double, line width=0.7pt] (24,4.5) -- (23,4.5) node[left] {$x$};
\draw[double, line width=0.7pt] (27,4.5) -- (28,4.5) node[right] {$a$};

\node at (25.5,3.5) {\large$M_{a|x}$};

\end{tikzpicture} .
\end{center}
To simplify the exposition, we presented the definition of the simplest super-channel: a comb with two legs. However, the definition naturally generalizes to multiple input and output channels. Using the Choi formalism, we show that the same decomposition theorem also holds in these cases, see Appendix \ref{app:SGHJW-proofs} for more details.
While the decomposition theorem for two-legged combs was already known in the literature \cite{Chiribella_2013-RN-supermaps} (even in the more general case of infinite-dimensional spaces), the generalisation to more complicated structures is our original contribution.

The careful reader may have noticed that all the results presented in this section share the exact same structure.
Consider an assemblage of a general quantum process $\{O_{a|x}\}_{a,x}$, it is said to be non-signalling if and only if
\begin{equation}\label{eq:ns_assemblage}
    \sum_a O_{a|x} = O_x = O \qquad \forall x.
\end{equation}
Assemblages satisfying this property can always be decomposed in a dilated quantum process of the same type, and a delayed measurement on the auxiliary space. Fig.~\ref{fig:decompositions} graphically illustrates this pattern, and a formal proof is presented in Appendix~\ref{app:SGHJW-proofs}.

\vspace{5pt}
\noindent\textbf{Communication scenarios and composability.}
Our interaction with the world is limited to observing correlations: probabilities of obtaining some classical outputs given some classical inputs.
Graphically, it means that experiments are diagrams without loose quantum wires.
Changing the nature of the internal wires, we can achieve different correlations. The most well studied cases are classical, quantum and non-signalling wires, where the latter encodes solely constraints from causality.
However, the introduction of communication often results in the collapse of this hierarchy.
For instance, in a two-players scenario with a one-way communication channel, a trivial classical protocol suffices to saturate the maximum correlations allowed by causality (the one-way non-signalling polytope). 

This raises the central question that motivates this letter: is no-communication necessary? Or can we impose restriction on communication scenarios such that the fundamental separations between classical, quantum and non-signalling correlation sets are preserved?

The S-G-HJW theorem for states finds a natural application in answering this question, by characterising the necessary and sufficient condition to relate non-local and communication scenarios.
This connection was known for the bipartite case, and properly formalised in \cite{Wright_2023}.
The proof becomes direct using the diagrams presented above.
If and only if the state assemblage is non-signalling, the following decomposition is valid:
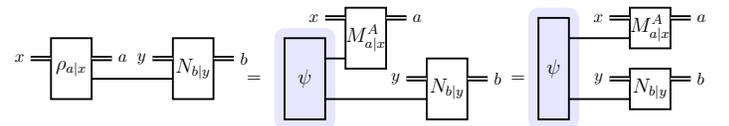
\vspace*{-1em}
\begin{center}
    \begin{tikzpicture}[scale=0.27, every node/.append style={scale=0.65}]

\begin{scope}[shift={(0,0)}]
    
\draw[line width=0.7pt] (0,0) -- (2,0) -- (2,3) -- (0,3) -- cycle;
\node at (1,1.5) {\large $\rho_{a|x}$};
\draw[double, line width=0.7pt] (0,2) -- (-1,2) node[left] {$x$};
    \draw[double, line width=0.7pt] (2,2) -- (3,2) node[right] {$a$};
\draw[line width=0.7pt] (2,1) -- (6,1);

\draw[line width=0.7pt] (6,0) -- (8,0) -- (8,3) -- (6,3) -- cycle;
\node at (7,1.5) {\large $N_{b|y}$};
\draw[double, line width=0.7pt] (6,2) -- (5,2) node[left] {$y$};
    \draw[double, line width=0.7pt] (8,2) -- (9,2) node[right] {$b$};

    \end{scope} 

\node at (10,1) {\large$=$};

\begin{scope}[shift={(11.5,-1)}]

\fill[col5!10, rounded corners] (-0.5,-0.5) -- (2.5,-0.5) -- (2.5,4.5) -- (-0.5,4.5) -- cycle;

\draw[line width=0.7pt] (0,0) -- (2,0) -- (2,4) -- (0,4) -- cycle;
\node at (1,2) {\large $\psi$};
\draw[line width=0.7pt] (2,3) -- (3,3);
\draw[line width=0.7pt] (2,1) -- (7,1);
\draw[line width=0.7pt] (3,2.5) -- (5,2.5) -- (5,5.5) -- (3,5.5) -- cycle;
\node at (4,4) {\large $M^A_{a|x}$};
\draw[double, line width=0.7pt] (3,5) -- (2,5) node[left] {$x$};
    \draw[double, line width=0.7pt] (5,5) -- (6,5) node[right] {$a$};

\node at (8,1.5) {\large $N_{b|y}$};

\draw[line width=0.7pt] (7,0) -- (9,0) -- (9,3) -- (7,3) -- cycle;
\draw[double, line width=0.7pt] (7,2) -- (6,2) node[left] {$y$};
    \draw[double, line width=0.7pt] (9,2) -- (10,2) node[right] {$b$};

    \end{scope} 

\node at (23,1) {\large$=$};

    \begin{scope}[shift={(24,-1)}]

\fill[col5!10, rounded corners] (-0.5,-0.5) -- (2,-0.5) -- (2,5.5) -- (-0.5,5.5) -- cycle;

    \draw[line width=0.7pt] (0,0) -- (1.5,0) -- (1.5,5) -- (0,5) -- cycle;
     \node at (0.75,2.5) {\large $\psi$};

    \draw[line width=0.7pt] (1.5,1) -- (4.5,1);
    \draw[line width=0.7pt] (1.5,4) -- (4.5,4);

    \draw[line width=0.7pt] (4.5,3.5) -- (6.5,3.5) -- (6.5,5.5) -- (4.5,5.5) -- cycle;
    \draw[line width=0.7pt] (4.5,0.5) -- (6.5,0.5) -- (6.5,2.5) -- (4.5,2.5) -- cycle;

    \draw[double, line width=0.7pt] (4.5,5) -- (3.5,5) node[left] {$x$};
    \draw[double, line width=0.7pt] (6.5,5) -- (7.5,5) node[right] {$a$};

        \draw[double, line width=0.7pt] (4.5,2) -- (3.5,2) node[left] {$y$};
    \draw[double, line width=0.7pt] (6.5,2) -- (7.5,2) node[right] {$b$};

    \node at (5.5,4.5) {\large $M^A_{a|x}$};
    \node at (5.5,1.5) {\large $N_{b|y}$};
    
\end{scope}

\end{tikzpicture}
 \end{center}
where the right hand figure is nothing but the bipartite Bell scenario -- meaning that every such prepare-and-measure scenario has a Bell bipartite realization.
Notice that in the last equality we simply made the first quantum wire longer; this operation is always allowed.

\begin{figure*}[t]
\centering
    \begin{subfigure}[t]{0.2\textwidth}
        \centering
\begin{tikzpicture}[scale=0.65, every node/.append style={scale=1},baseline=(current bounding box.center)]{

\node[draw, minimum size=0.7cm] (A) at (0,0) {};
\node[draw, minimum size=0.7cm] (B) at (2,0) {};
\node[draw, minimum size=0.7cm] (C) at (5,0) {};

\draw[->, very thick] (A) -- (B);
\draw[->, very thick] (B) -- (3,0);
\draw[->, very thick] (4,0) -- (C);
\draw[loosely dotted, very thick] (3.2,0) -- (3.8,0);

\draw[->] ($(A.north)+(0,0.4)$) node[above] {$x_1$} -- (A.north);
\draw[->] ($(B.north)+(0,0.4)$) node[above] {$x_2$} -- (B.north);
\draw[->] ($(C.north)+(0,0.4)$) node[above] {$x_\mathsf{k}$} -- (C.north);

\draw[->] (A.south) -- ++(0,-0.4) node[below] {$a_1$};
\draw[->] (B.south) -- ++(0,-0.4) node[below] {$a_2$};
\draw[->] (C.south) -- ++(0,-0.4) node[below] {$a_\mathsf{k}$};
 };
\end{tikzpicture}
\caption{$\mathsf{k}$-sequential.}\label{fig:k-seq}
    \end{subfigure}
    \hfill
    \begin{subfigure}[t]{0.25\textwidth}
        \centering

\begin{tikzpicture}[scale=0.65, every node/.append style={scale=1},baseline=(current bounding box.center)]
{
\node[draw, minimum size=0.7cm] (A1) at (-1,1) {$1$};
\node[draw, minimum size=0.7cm] (A2) at (0,3) {$2$};

\node[draw, minimum size=0.7cm] (B1) at (2.5,2){$3$};

\node[draw, minimum size=0.7cm] (C1) at (5,0) {$6$};
\node[draw, minimum size=0.7cm] (C2) at (6,2) {$5$};
\node[draw, minimum size=0.7cm] (C3) at (5,4) {$4$};

\draw[->, very thick] (A1) -- (B1);
\draw[->, very thick] (A2) -- (B1);

\draw[->, very thick] (B1) -- (C1);
\draw[->, very thick] (B1) -- (C2);
\draw[->, very thick] (B1) -- (C3);

\draw[->, very thick] (A2) -- (C3);

\draw[->] ($(A1.north)+(0,0.4)$) node[above] {$x_1$} -- (A1.north);
\draw[->] ($(A2.north)+(0,0.4)$) node[above] {$x_2$} -- (A2.north);
\draw[->] ($(B1.north)+(0,0.4)$) node[above] {$x_3$} -- (B1.north);
\draw[->] ($(C1.north)+(0,0.4)$) node[above] {$x_6$} -- (C1.north);
\draw[->] ($(C2.north)+(0,0.4)$) node[above] {$x_5$} -- (C2.north);
\draw[->] ($(C3.north)+(0,0.4)$) node[above] {$x_4$} -- (C3.north);

\draw[->] (A1.south) -- ++(0,-0.4) node[below] {$a_1$};
\draw[->] (A2.south) -- ++(0,-0.4) node[below] {$a_2$};
\draw[->] (B1.south) -- ++(0,-0.4) node[below] {$a_3$};
\draw[->] (C1.south) -- ++(0,-0.4) node[below] {$a_6$};
\draw[->] (C2.south) -- ++(0,-0.4) node[below] {$a_5$};
\draw[->] (C3.south) -- ++(0,-0.4) node[below] {$a_4$};
 };
\end{tikzpicture}
\caption{Directed acyclic graph.}\label{fig:dag}

    \end{subfigure}
    \hfill
    \begin{subfigure}[t]{0.2\textwidth}
        \centering

\begin{tikzpicture}[scale=0.65, every node/.append style={scale=1},baseline=(current bounding box.center)]{
\fill[gray!30] (0,0) -- (0.5,0) --(0.5,6.5) -- (0,6.5) -- cycle;
    \fill[gray!30] (3.5,0) -- (4,0) --(4,6.5) -- (3.5,6.5) -- cycle;

    \fill[gray!30] (0,3) -- (4,3) --(4,3.5) -- (0,3.5) -- cycle;

    \node[draw, minimum size=0.7cm] (A) at (2,1.5) {};
   \node[draw, minimum size=0.7cm] (B) at (2,5) {};

   \draw[->] ($(A.north)+(0,0.4)$) node[above] {$x$} -- (A.north);
    \draw[->] ($(B.north)+(0,0.4)$) node[above] {$y$} -- (B.north);
    \draw[->] (A.south) -- ++(0,-0.4) node[below] {$a$};
    \draw[->] (B.south) -- ++(0,-0.4) node[below] {$b$};

    \draw[->, very thick] (0.5,1.5) -- (A);
    \draw[->, very thick] (A) -- (3.5,1.5);

    \draw[->, very thick] (0.5,5) -- (B);
    \draw[->, very thick] (B) -- (3.5,5);

 };
\end{tikzpicture}
\caption{Indefinite causal order.}\label{fig:ico}

    \end{subfigure}
    \hfill
    \begin{subfigure}[t]{0.18\textwidth}
        \centering
\begin{tikzpicture}[scale=0.65, every node/.append style={scale=1},baseline=(current bounding box.center)]{

\node[draw, minimum size=0.7cm] (A) {};
\node[draw, minimum size=0.7cm, right=of A] (B) {};

\draw[->, very thick] ($(A.east)+(0,0.2)$) to[bend left=20] ($(B.west)+(0,0.2)$);
\draw[->, very thick] ($(B.west)-(0,0.2)$) to[bend left=20] ($(A.east)-(0,0.2)$);

\draw[->] ($(A.north)+(0,0.4)$) node[above] {$x$} -- (A.north);
\draw[->] ($(B.north)+(0,0.4)$) node[above] {$y$} -- (B.north);

\draw[->] (A.south) -- ++(0,-0.4) node[below] {$a$};
\draw[->] (B.south) -- ++(0,-0.4) node[below] {$b$};
 };
\end{tikzpicture}
\caption{The simplest loop.}\label{fig:loop}

    \end{subfigure} \caption{\justifying Different communication schemes. Each party is represented by a square, the thin vertical arrows are classical inputs and outputs, and the thick arrows are quantum channels. The grey `H' in (c) represents a process matrix \cite{Oreshkov_2012}.
}\label{fig:compositions-overview}
\end{figure*}
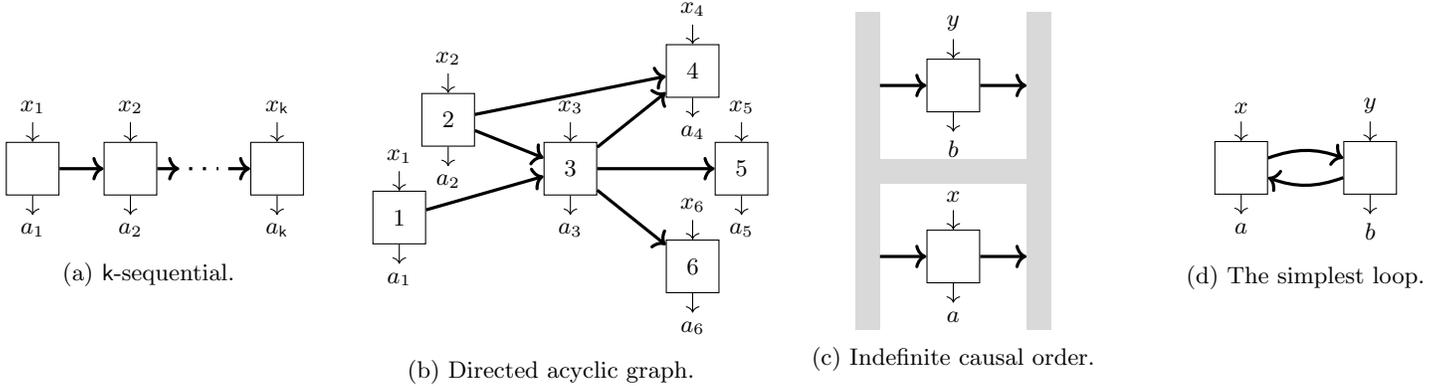

A similar result is not known for any other communication scenario.
Consider, for example, three sequential players connected by a one-way quantum channel, a prepare-transform-measure scenario.
Following the bipartite proof, one might expect that if all the states exchanged by the parties are non-signalling state assemblages, this would suffice to reconstruct a tripartite Bell realisation.
The S-G-HJW theorem allows to construct two separate bipartite models, one for each set of non-signalling state assemblages; in the first, the middle player’s action is grouped with the state preparation, in the second, it is grouped with the measurement.
However it is unclear how to combine these to obtain a Bell tripartite model.
Even more, the existence of post-quantum steering \cite{Sainz_2015} seems to suggest that this is not possible. The bottleneck is that non-signalling assemblages of states, and their decomposition,
are not-composable.

\setlength{\parskip}{0pt}
To overcome this, we propose to model the intermediate player as a quantum instrument, and consider non-signalling assemblages of instruments instead.
Then, we can sequentially apply the state decomposition for the first player, and the instrument decomposition for the second:
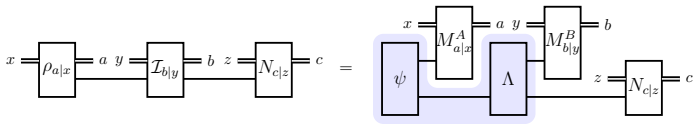
\begin{center}
\begin{tikzpicture}[scale=0.24, every node/.append style={scale=0.6}]

\begin{scope}[shift={(0,0)}]
    
\draw[line width=0.7pt] (0,0) -- (2,0) -- (2,3) -- (0,3) -- cycle;
\node at (1,1.5) {\large $\rho_{a|x}$};
\draw[double, line width=0.7pt] (0,2) -- (-1,2) node[left] {$x$};
    \draw[double, line width=0.7pt] (2,2) -- (3,2) node[right] {$a$};
\draw[line width=0.7pt] (2,1) -- (6,1);

\draw[line width=0.7pt] (6,0) -- (8,0) -- (8,3) -- (6,3) -- cycle;
\node at (7,1.5) {\large $\mathcal{I}_{b|y}$};
\draw[double, line width=0.7pt] (6,2) -- (5,2) node[left] {$y$};
    \draw[double, line width=0.7pt] (8,2) -- (9,2) node[right] {$b$};
\draw[line width=0.7pt] (8,1) -- (12,1);

\draw[line width=0.7pt] (12,0) -- (14,0) -- (14,3) -- (12,3) -- cycle;
\node at (13,1.5) {\large $N_{c|z}$};
\draw[double, line width=0.7pt] (12,2) -- (11,2) node[left] {$z$};
    \draw[double, line width=0.7pt] (14,2) -- (15,2) node[right] {$c$};
    
    \end{scope} 

\node at (17,1.5) {\large$=$};

\begin{scope}[shift={(19,-1)}]

\fill[col5!10, rounded corners] (-0.5,-0.5) -- (8.5,-0.5) -- (8.5,4.5) -- (5.5,4.5) -- (5.5,1.5) -- (2.5,1.5) -- (2.5,4.5) -- (-0.5,4.5) -- cycle;

\draw[line width=0.7pt] (0,0) -- (2,0) -- (2,4) -- (0,4) -- cycle;
\node at (1,2) {\large $\psi$};
\draw[line width=0.7pt] (2,3) -- (3,3);
\draw[line width=0.7pt] (2,1) -- (6,1);
\draw[line width=0.7pt] (3,2) -- (5,2) -- (5,6) -- (3,6) -- cycle;
\node at (4,4) {\large $M^A_{a|x}$};
\draw[double, line width=0.7pt] (3,5) -- (2,5) node[left] {$x$};
    \draw[double, line width=0.7pt] (5,5) -- (6,5) node[right] {$a$};

\draw[line width=0.7pt] (6,0) -- (8,0) -- (8,4) -- (6,4) -- cycle;
\node at (7,2) {\large $\Lambda$};
\draw[line width=0.7pt] (8,3) -- (9,3);
\draw[line width=0.7pt] (8,1) -- (13.5,1);
\draw[line width=0.7pt] (9,2) -- (11,2) -- (11,6) -- (9,6) -- cycle;
\node at (10,4) {\large $M^B_{b|y}$};
\draw[double, line width=0.7pt] (9,5) -- (8,5) node[left] {$y$};
    \draw[double, line width=0.7pt] (11,5) -- (12,5) node[right] {$b$};

\draw[line width=0.7pt] (13.5,0) -- (15.5,0) -- (15.5,3) -- (13.5,3) -- cycle;
\node at (14.5,1.5) {\large $N_{c|z}$};
\draw[double, line width=0.7pt] (13.5,2) -- (12.5,2) node[left] {$z$};
    \draw[double, line width=0.7pt] (15.5,2) -- (16.5,2) node[right] {$c$};

    \end{scope} 

\end{tikzpicture} \end{center}
where the state $\psi$ followed by the isometry $\Lambda$ can be interpreted as a tripartite resource. Therefore, the right-hand side of the diagram represents a tripartite Bell scenario.
The same reasoning extends recursively to a sequence of $\mathsf{k}$ players, all implementing non-signalling composable processes. The proof structure is always the same: apply the decomposition theorem to each non-signalling assemblage, delay the measurements and consider the purifications generated by the decomposition as one big source preparation.

\setlength{\parskip}{0pt}
Let us now present in more detail more complex communication scenarios, where this result also applies.
It is common to represent these graphically with a different notation, that we introduce in Fig.~\ref{fig:compositions-overview}.

To go beyond the sequential case, consider players that can receive multiple inputs and output multiple quantum systems. Such a structure can be generally described with a directed acyclic graph (DAG, as in Fig.~\ref{fig:dag}) of $\textsf{n}$ players.
As long as each preparation and transformation assemblage is non-signalling, the resulting correlations always have an $\textsf{n}$-partite Bell counterpart.

Even if the communication pattern is not described by a definite causal order (Fig.~\ref{fig:ico}), the same result holds. That is, if we consider correlations $p(ab|xy)=\Tr[W (\Tilde{\mathcal{I}}^A_{a|x} \otimes \Tilde{\mathcal{I}}^B_{b|y})] $, where $W$ is a process matrix~\cite{Oreshkov_2012}, and $\Tilde{\mathcal{I}}^A_{a|x}$ and $\Tilde{\mathcal{I}}^B_{b|y}$ are the Choi operators of a non-signalling instrument assemblages, the correlations is non-signalling~\cite{bavaresco2024indefinitecausalorderboxworld}. 
Moreover, we show that these correlations form precisely the set of quantum spatial correlations.
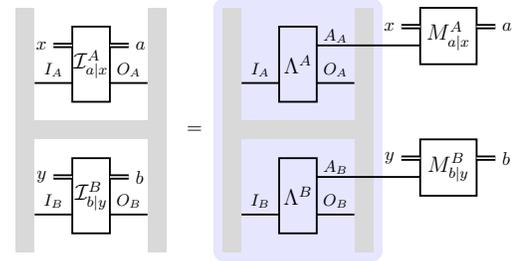
\begin{center}
         \centering
\begin{tikzpicture}[scale=0.25, every node/.append style={scale=0.7}]
\fill[gray!30] (0,0) -- (1,0) -- (1,13) -- (0,13) -- cycle;
\fill[gray!30] (7,0) -- (8,0) -- (8,13) -- (7,13) -- cycle;
\fill[gray!30] (0,6) -- (8,6) -- (8,7) -- (0,7) -- cycle;

\draw[line width=0.7pt] (3,1) -- (5,1) -- (5,5) -- (3,5) -- cycle;

\draw[line width=0.7pt] (3,8) -- (5,8) -- (5,12) -- (3,12) -- cycle;

  \draw[line width=0.7pt] (1,9) -- (3,9) node[midway, above] {\small$I_A$};
  \draw[line width=0.7pt] (5,9) -- (7,9) node[midway, above] {\small$O_A$};
  \draw[line width=0.7pt] (1,2) -- (3,2) node[midway, above] {\small$I_B$};
  \draw[line width=0.7pt] (5,2) -- (7,2) node[midway, above] {\small$O_B$};

  \draw[double, line width=0.7pt] (3,11) -- (2,11) node[left] {$x$};
  \draw[double, line width=0.7pt] (5,11) -- (6,11) node[right] {$a$};

  \draw[double, line width=0.7pt] (3,4) -- (2,4) node[left] {$y$};
  \draw[double, line width=0.7pt] (5,4) -- (6,4) node[right] {$b$};

    \node at (9.5, 6.5) {\large$=$};

\fill[col5!10, rounded corners] (10.5,-0.5) -- (19.5, -0.5) -- (19.5, 13.5) -- (10.5,13.5) -- cycle;

\fill[gray!30] (11,0) -- (12,0) -- (12,13) -- (11,13) -- cycle;
\fill[gray!30] (18,0) -- (19,0) -- (19,13) -- (18,13) -- cycle;
\fill[gray!30] (11,6) -- (19,6) -- (19,7) -- (11,7) -- cycle;

\draw[line width=0.7pt] (14,1) -- (16,1) -- (16,5) -- (14,5) -- cycle;
\draw[line width=0.7pt] (14,8) -- (16,8) -- (16,12) -- (14,12) -- cycle;

\draw[line width=0.7pt] (12,9) -- (14,9) node[midway, above] {\small$I_A$};
\draw[line width=0.7pt] (16,9) -- (18,9) node[midway, above] {\small$O_A$};
\draw[line width=0.7pt] (12,2) -- (14,2) node[midway, above] {\small$I_B$};
\draw[line width=0.7pt] (16,2) -- (18,2) node[midway, above] {\small$O_B$};

\draw[line width=0.7pt] (16,11) -- (21.5,11);
\draw[line width=0.7pt] (16,4) -- (21.5,4);

\draw[line width=0.7pt] (21.5,3) -- (24.5,3) -- (24.5,6) -- (21.5,6) -- cycle;
\draw[line width=0.7pt] (21.5,10) -- (24.5,10) -- (24.5,13) -- (21.5,13) -- cycle;

\draw[double, line width=0.7pt] (21.5,12) -- (20.5,12) node[left] {$x$};
\draw[double, line width=0.7pt] (24.5,12) -- (25.5,12) node[right] {$a$};

\draw[double, line width=0.7pt] (21.5,5) -- (20.5,5) node[left] {$y$};
\draw[double, line width=0.7pt] (24.5,5) -- (25.5,5) node[right] {$b$};

\node at (4,10) {\large$\mathcal{I}^A_{a|x}$};
\node at (4,3) {\large$\mathcal{I}^B_{b|y}$};

\node at (15,10) {\large$\Lambda^A$};
\node at (15,3) {\large$\Lambda^B$};

\node at (23,11.5) {\large$M^A_{a|x}$};
\node at (23,4.5) {\large$M^B_{b|y}$};

\node at (17,11.5) {\small$A_A$};
\node at (17,4.5) {\small$A_B$};
  
\end{tikzpicture}

     \end{center}
Despite their generality, process matrices still do not account for loops of communication. For instance, consider the simplest loop structure of Fig.~\ref{fig:loop}, where two players are allowed a round of back-and-forth communication, before committing to a classical output.
The most general strategy for the first player is to prepare a bipartite state conditioned on $x$, send a register to the second player and keep the other as a quantum memory, and finally perform a joint measurement on the state she receives and her memory; in our diagrammatic language it becomes evident that this is a super-instrument.
Fortunately, the concept of non-signalling assemblage extends to this object, and we have a corresponding decomposition:
    
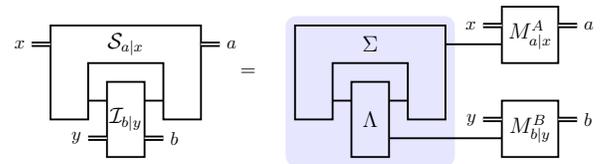
\begin{center}
        \centering
\begin{tikzpicture}[scale=0.25, every node/.append style={scale=0.7}]
\draw[line width=0.7pt] (1,0) -- (3,0) -- (3,3) -- (7,3) -- (7,0) -- (9,0) --(9,5) -- (1,5) -- cycle;

  \draw[double, line width=0.7pt] (1,4) -- (0,4) node[left] {$x$};
  \draw[double, line width=0.7pt] (9,4) -- (10,4) node[right] {$a$};

  \draw[line width=0.7pt] (3,1) -- (4,1) ;
  \draw[line width=0.7pt] (6,1) -- (7,1);

  \node at (5,4) {\large$\mathcal{S}_{a|x}$};

    \draw[line width=0.7pt] (4,-2) -- (6,-2) -- (6,2) -- (4,2) -- cycle;

      \node at (5,0) {\large$\mathcal{I}_{b|y}$};

  \draw[double, line width=0.7pt] (4,-1) -- (3,-1) node[left] {$y$};
  \draw[double, line width=0.7pt] (6,-1) -- (7,-1) node[right] {$b$};

    \node at (11.5,2.5) {\large$=$};

\fill[col5!10, rounded corners] (13.5,-2.5) -- (22.5, -2.5) -- (22.5, 5.5) -- (13.5,5.5) -- cycle;

\draw[line width=0.7pt] (14,0) -- (16,0) -- (16,3) -- (20,3) -- (20,0) -- (22,0) -- (22,5) -- (14,5) -- cycle;

\draw[line width=0.7pt] (16,1) -- (17,1);
\draw[line width=0.7pt] (19,1) -- (20,1);

\node at (18,4) {\large$\Sigma$};

\draw[line width=0.7pt] (22,4) -- (25,4);\draw[line width=0.7pt] (25,3) -- (28,3) -- (28,6) -- (25,6) -- cycle;
\draw[double, line width=0.7pt] (25,5) -- (24,5) node[left] {$x$};
\draw[double, line width=0.7pt] (28,5) -- (29,5) node[right] {$a$};
\node at (26.5,4.5) {\large$M^A_{a|x}$};

\draw[line width=0.7pt] (17,-2) -- (19,-2) -- (19,2) -- (17,2) -- cycle;
\node at (18,0) {\large$\Lambda$};

\draw[line width=0.7pt] (19,-1) -- (25,-1);
\draw[line width=0.7pt] (25,-2) -- (28,-2) -- (28,1) -- (25,1) -- cycle;
\draw[double, line width=0.7pt] (25,0) -- (24,0) node[left] {$y$};
\draw[double, line width=0.7pt] (28,0) -- (29,0) node[right] {$b$};
\node at (26.5,-0.5) {\large$M^B_{b|y}$};

\end{tikzpicture}     \end{center}
To consider more complicated communication loops, it suffices to consider more complex super-instruments (see Appendix~\ref{app:example-comm-scenarios}). 
These are capturing all possible communication scenarios.
Similar ideas are known as multi-round process matrices \cite{Hoffreumon_2021}.

\begin{theorem}\label{thm:general-composition}
    Let $\mathsf{k} \in \mathbb{N}$. The arbitrary composition of $\mathsf{k}$ non-signalling quantum assemblages
always produces correlations which admit a $\mathsf{k}$-partite Bell quantum model.
\end{theorem}
\noindent As the contraposition, it follows that any interactive quantum realisation of $\mathsf{k}$-partite correlations that are outside the Bell quantum set must involve signalling communication.
This
holds even if the no-signalling principle is satisfied at the level of correlations,
and can be seen as a generalisation of the well-established fact that signalling is necessary to realise PR box correlations in one-way communication scenarios~\cite{PR}.

An analogous classical result follows immediately by replacing the non-signalling decompositions with the unsteerable decompositions \cite{Piani_2015,Uola_2018} depicted in Fig.~\ref{fig:decompositions}~(iii),
given by
\begin{equation}\label{eq:uns_assemblage}
    O_{a|x} = \sum_{\lambda} p(\lambda) p(a|x,\lambda) \Omega_\lambda,
\end{equation}
where $\{O_{a|x}\}_{a,x}$ is an assemblage, $\{\Omega_\lambda\}_\lambda$ is a family of (normalized) quantum objects of the same type, and $\lambda$ describes the classical variable on the internal wire.
\begin{theorem}\label{thm:general-composition-classical}
    Let $\mathsf{k} \in \mathbb{N}$. The arbitrary composition of $\mathsf{k}$ unsteerable assemblages always produces correlations which admit a $\mathsf{k}$-partite Bell local realisation.
\end{theorem}

\vspace{5pt}
\noindent\textbf{Related Works and Discussion.}
Given the versatility of the S-G-HJW theorem, similar questions concerning its composition have been formulated in multiple areas of quantum information.
In this section, we outline some of these contexts, and explain how our result relates to them.

\textit{Channel steering and post-quantum steering.}
The definition of non-signalling assemblage of instruments is equivalent to the non-signalling channels introduced by Piani in \cite{Piani_2015}, and reconsidered in \cite{Uola_2018}, where they also report the decomposition theorem in its operator-algebraic form.
However, they do not consider the composition of multiple instruments, nor the sequential applications of the theorem.
Furthermore, our work has evident connections to post-quantum steering \cite{Sainz_2015}, for which the topic of interest are non-signalling bipartite assemblages of states:
  $  \sum_a \sigma_{ab|xy} = \sigma_{b|y},
    \sum_b \sigma_{ab|xy} = \sigma_{a|x}.$

Ref.~\cite{Sainz_2015} ask a question very related to ours: can every non-signalling bipartite state assemblage  be prepared by a tripartite spatial Bell quantum model, \emph{i.e.}, $\sigma_{ab|xy}=\Tr_{AB}(\rho_{ABC} \; A_{a|x}\otimes B_{b|y}\otimes \mathds{1} )$, where $\rho_{ABC}$ is a quantum state and $\{A_{a|x}\}_{ax},\{B_{b|y}\}_{by}$ are POVMs?
The answer is no. Maybe this is less surprising after our results, that characterises the necessary and sufficient conditions to connect sequential and non-local scenarios. Nevertheless our theorem cannot be directly applied in this setting, because in these bipartite assemblages we cannot separate Alice's and Bob's labels.

\textit{Contextuality and GPTs.}
This work can be seen as the first step for a natural extension of the map presented in \cite{Wright_2023} to multipartite scenarios.
What we call non-signalling assemblages of states and instruments, in the generalised contextuality framework \cite{Spekkens_2005} are called preparation and transformation equivalences.
We focused on the quantum case, but by considering the decomposition of Fig~\ref{fig:decompositions} as the definition of non-signalling assemblage (ii) and unsteerable assemblage (iii), these theorems and their compositions extend to generalised probabilistic theories (GPTs)~\cite{barrett2005GPT,plavala2023GPT} and operational probabilistic theories (OPT) with purification~\cite{chiribella2010opt}. For instance, using the definitions of state, channel, instruments, and measurements from \cite{plavala2023GPT} or \cite{chiribella2010opt}, it follows  that composition of non-signalling assemblages respecting the decomposition\footnote{From a quantum perspective, our results ensure the equivalence between two definitions of non-signalling assemblages, an algebraic one, given by Eq.~\eqref{eq:ns_assemblage}, and a more operational one, given by the decomposition (ii) of Fig.~\ref{fig:decompositions}. In arbitrary GPTs and OPTs, the equivalence by these two definitions may not hold, hence, in order to make sure Thm.~\ref{thm:general-composition} generalises to arbitrary GPTs and OPTs, we may simply set the decomposition (ii) of Fig.~\ref{fig:decompositions} as the definition of a non-signalling assemblage. Analogously, we may define unsteerable assemblage as the decomposition (iii) of Fig.~\ref{fig:decompositions} instead of the algebraic definition of Eq.~\eqref{eq:uns_assemblage}, to ensure that Thm.~\ref{thm:general-composition-classical} generalises to GPTs and OPTs.} (ii) illustrated in Fig.~\ref{fig:decompositions} will necessarily lead to spatial correlations. And, conversely, all spatial correlations can be realised in the corresponding communication scenario.

\textit{Infinite dimensions and Radon-Nikodym theorems.}
In this letter, Hilbert spaces were assumed to be finite-dimensional, a common assumption in quantum information. 
In full generality, one should consider the operator-algebraic framework, that allows to neatly treat infinite dimensions and sub-algebras of $\mathsf{B}(\mathcal{H})$.
The S-G-HJW theorem has a natural algebraic equivalent, the celebrated Radon-Nikodym theorem for positive linear functionals. Interestingly, before this work, the higher-order generalisations to instruments \cite{Raginsky} and super-instruments \cite{Chiribella_2013-RN-supermaps} were mainly formulated in this language, and perhaps because of this less known in the quantum physics community.
However, the composition results do not translate directly.
A chain rule for the Radon-Nikodym theorem was recently proven in the context of compiled nonlocal games \cite{baroni2025boundingasymptoticquantumvalue}. This corresponds in this work to the case of $\mathsf{k}$ sequential players. We leave as an open question whether an analogous statement holds for other communication scenarios in infinite dimensions.

\vspace{5pt}
\textbf{Acknowledgments.}
MB and DL would like to thank Rob Spekkens for discussions about the naming of the no-signalling property for assemblages.
MB acknowledges funding from QuantEdu France, a State aid managed by the French National Research Agency for France 2030 with the reference ANR-22-CMAS-0001.
DL acknowledges support from the Quantum Advantage Pathfinder (QAP) research program within the UK’s National Quantum Computing Center (NQCC).
MTQ is supported by the French Agence Nationale de la Recherche (ANR) under grant JCJC HOQO-KS.
MB and DM acknowledge funding from the PEPR integrated project EPIQ ANR-22-PETQ-0007, part of Plan France 2030.

\nocite{apsrev42Control} 
\bibliographystyle{apsrev4-2}
\bibliography{biblio}

\clearpage
\appendix \setcounter{secnumdepth}{3}
\renewcommand{\thesection}{\Alph{section}}
\renewcommand{\thesubsection}{\thesection\arabic{subsection}}
\setcounter{equation}{0}
\renewcommand{\theequation}{\thesection\arabic{equation}}

\section{Simultaneous purification of non-signalling assemblages of arbitrary quantum objects} \label{app:SGHJW-proofs}

In this section, we state the most general version of the simultaneous purification results presented in this work, and provide its proof.
We show that the non-signalling property of assemblages is a necessary and sufficient condition for decomposition~(ii) from Figure~\ref{fig:decompositions} not only for assemblages of states, instruments, or super-instruments, but also for arbitrary quantum objects.
To the best of our knowledge, this fact has not been formally recorded previously.

We start by stating again the well-known simultaneous purification theorem for states, for which many constructive proofs can be found in the literature~\cite{Schrödinger_1936,gisin1989stochastic,hughston1993complete,kirkpatrick2005schrodingerhjwtheorem}.
In the following, we will reduce the general case of quantum objects to this theorem.
The advantage of this reduction is that it provides a concrete and constructive way of finding such a simultaneous purification, also in the general case, whenever it exists.
\begin{theorem}\label{thm:simultaneous-purification-states}
    Let $\{\rho_{a|x}\}_{a,x} \subseteq \mathsf{B}(\mathcal{H})$ be an assemblage of quantum states on a finite-dimensional Hilbert space $\mathcal{H}$. Then, the following two conditions are equivalent:
    \begin{enumerate}
        \item The assemblage $\{\rho_{a|x}\}_{a,x}$ is no-signalling, \emph{i.e.}, there exists a state $\rho \in \mathsf{B}(\mathcal{H})$, such that
            \begin{equation}
                \sum_a \rho_{a|x} = \rho, \quad \text{for all } x.
            \end{equation}
        \item There exists a (finite-dimensional) auxiliary Hilbert space $\mathcal{H}_A$, a pure quantum state $\ket{\psi} \in \mathcal{H}_A \otimes \mathcal{H}$, and a family of measurements $\{M_{a|x}\}_{a,x} \subseteq \mathsf{B}(\mathcal{H}_A)$, such that
            \begin{equation}
                \rho_{a|x} = \Tr_A \left[ \left(M_{a|x} \otimes \mathds{1} \right) \ketbra{\psi}{\psi} \right].
            \end{equation}
    \end{enumerate}
\end{theorem}
Note that Theorem~\ref{thm:simultaneous-purification-states} remains true even if the assemblage is not normalized, in which case also the states $\rho$ and $\ket{\psi}$ become non-normalized states.
\begin{proof}
Let us report a fairly standard construction, which can be found for example in~\cite{Navascu_s_2012}, that works for any separable Hilbert space $\mathcal{H}$.
    Let $\mathcal{H}_A$ be a copy of $\mathcal{H}$ and define the operators
    \begin{align}
         M_{a|x} := \left[ \rho^{-1/2} \rho_{a|x} \rho^{-1/2}\right]^T,
    \end{align}
    on $\mathcal{H}_A$, where the transpose is with respect to some fixed orthonormal basis $\ket{i}$ of $\mathcal{H}_A$, the positive square root is always well-defined for positive operators, and the inverse is the Moore-Penrose pseudo-inverse.
    Clearly all operators $M_{a|x}$ are positive and it holds that
    \begin{equation}
        \sum_a M_{a|x} = \left[ \rho^{-1/2} \left( \sum_a \rho_{a|x} \right)  \rho^{-1/2}\right]^T = \mathds{1}.
    \end{equation}
    Hence, the operators $\{M_{a|x}\}_{a,x}$ form a family of POVMs.
    Now, define
    \begin{align}
         \ket{\psi} := \sum_{i} \ket{i} \otimes \rho^{1/2} \ket{i} \; \in \mathcal{H}_A \otimes \mathcal{H}.
    \end{align}
    Then, $\braket{\psi|\psi} = \Tr [ \rho ] = 1$, and therefore $\ket{\psi}$ describes a normalized, pure quantum state. A direct calculation verifies that
    \begin{align}
        \Tr_A \left[ \left(M_{a|x} \otimes \mathds{1} \right) \ketbra{\psi}{\psi} \right] = \rho_{a|x},
    \end{align}
    which concludes the proof.
\end{proof}

This theorem can be generalised from assemblages of states to assemblages of quantum instruments \cite{Raginsky, Uola_2018}.
This can be seen as a direct consequence of a famous theorem in C$^*$-algebras: the Radon-Nikodym theorem for CP maps \cite{BS86radon}.
However, if we restrict the dimension of the Hilbert spaces to be finite, we can avoid the C$^*$-algebraic proof, and retrieve the same result using Kraus or Choi operators.

\begin{theorem}\label{thm:simultaneous-purification-instruments}
    Let $\{\mathcal{I}_{a|x}\}_{a,x} \subseteq \mathsf{CP}(\mathcal{H}_I, \mathcal{H}_O)$ be an assemblage of instruments on finite-dimensional Hilbert spaces $\mathcal{H}_I$ and $\mathcal{H}_O$. Then, the following two conditions are equivalent:
    \begin{enumerate}
        \item The assemblage $\{\mathcal{I}_{a|x}\}_{a,x}$ is no-signalling, \emph{i.e.}, there exists a channel $\mathcal{I} \in \mathsf{CPTP}(\mathcal{H}_I, \mathcal{H}_O)$, such that
            \begin{equation}
                \sum_a \mathcal{I}_{a|x} = \mathcal{I}, \quad \text{for all } x.
            \end{equation}
        \item There exists a (finite-dimensional) auxiliary Hilbert space $\mathcal{H}_A$, an isometric channel $\Lambda(\cdot)= V (\cdot) V^\dagger \in \mathsf{CPTP}(\mathcal{H}_I, \mathcal{H}_A \otimes \mathcal{H}_O)$, with $V$ an isometry, and a family of measurements $\{M_{a|x}\}_{a,x} \subseteq \mathsf{B}(\mathcal{H}_A)$, such that
            \begin{equation}
                \mathcal{I}_{a|x}(\rho)= \Tr_A \left[ \left(M_{a|x} \otimes \mathds{1} \right) \Lambda(\rho) \right].
            \end{equation}
    \end{enumerate}
\end{theorem}
\begin{proof}
We present two equivalent constructive proofs, using Kraus and Choi operators, respectively.

\textbf{Kraus proof.}
Fix a minimal Kraus decomposition for $\mathcal{I}$
\begin{equation}
    \mathcal{I}(\rho)  = \sum_{i=1}^r K_i \rho K_i^\dagger,
    \qquad \sum_i K_i^\dagger K_i = \mathds{1},
\end{equation}
and let $\mathcal{H}_A$ be an $r$-dimensional Hilbert space with basis $\{\ket{i}\}_{i=1,\dots,r}$.
The associated Stinespring isometry $V : \mathcal{H}_I \to \mathcal{H}_A \otimes \mathcal{H}_O$ is then given by
\begin{equation}
    V : \ket{v} \mapsto \sum_{i=1}^r \ket{i} \otimes K_i \ket{v},
\end{equation}
so that the isometric channel is 
\begin{equation}
    \Lambda(\rho)= V \rho V^\dagger = \sum_{i,j}\ketbra{i}{j} \otimes  K_i \rho K_j^\dagger.
\end{equation}
We can similarly define a not-necessarily minimal Kraus decomposition for all CP maps $\mathcal{I}_{a|x}$ by
\begin{equation}
    \mathcal{I}_{a|x}(\rho)  = \sum_{\mu} L_\mu^{(a,x)} \rho \left( L_\mu^{(a,x)} \right)^\dagger,
\end{equation}
such that
\begin{equation}
    \sum_\mu \left( L_\mu^{(a,x)} \right)^\dagger \! L_\mu^{(a,x)}\leq \mathds{1}, \;\; \sum_a \sum_\mu \left( L_\mu^{(a,x)} \right)^\dagger \! L_\mu^{(a,x)} = \mathds{1}.
\end{equation}
By assumption we know that $\sum_a \mathcal{I}_{a|x} = \mathcal{I}$, hence for any fixed $x$, the Kraus families $\{L_\mu^{(a,x)}\}_{\mu,a}$ must generate the same map as $\{K_i\}_i$.
Consequently, by the minimality of the Kraus decomposition $\{K_i\}_i$, all operators $L_\mu^{(a,x)}$ must be expressible as a linear combination of the $K_i$'s in the following way
\begin{equation}
    L_\mu^{(a,x)} = \sum_{i=1}^r u_{\mu,i}^{(a,x)} K_i, \qquad u_{\mu,i}^{(a,x)} \in \mathbb{C},
\end{equation}
where the coefficients $u_{\mu,i}^{(a,x)}$ form an isometry $U^{(x)}$ in the sense that they satisfy
\begin{equation}\label{eqn:kraus-isometry}
    \sum_{a,\mu} \left( u_{\mu,i}^{(a,x)} \right)^* u_{\mu,j}^{(a,x)} = \left\{ \begin{matrix}
        1, & \text{ if } i=j, \\
        0, & \text{ otherwise.}
    \end{matrix}
    \right.
\end{equation}
Now, the vectors
\begin{equation}
    \ket{u_{\mu}^{(a,x)}}= \sum_{i=1}^r u_{\mu,i}^{(a,x)} \ket{i} \; \in \mathcal{H}_A
 \end{equation}
yield by \eqref{eqn:kraus-isometry} a family of POVMs $\{M_{a|x}\}_{a,x}$ on the ancilla space by letting
\begin{equation}
    M_{a|x} = \sum_\mu \ketbra{u_{\mu}^{(a,x)}}{u_{\mu}^{(a,x)}}.
\end{equation}
Finally, it is easily checked that
\begin{align}
     &\Tr_A \left[ (M_{a|x} \otimes \mathds{1}) V \rho V^\dagger \right] \nonumber \\
     &=\Tr_A \left[ \sum_{i,j,\mu} \ketbra{u_{\mu}^{(a,x)}}{u_{\mu}^{(a,x)}} \ketbra{i}{j} \otimes K_i \rho K_j^\dagger
     \right] \\
     &= \sum_{i,j,\mu} u_{\mu,j}^{(a,x)} \bar{u}_{\mu,i}^{(a,x)} K_i \rho K_j^\dagger
     = \sum_{\mu} L_\mu^{(a,x)} \rho \left( L_\mu^{(a,x)} \right)^\dagger
     =\mathcal{I}_{a|x}(\rho). \nonumber
\end{align}

\textbf{Choi proof.}
As the Choi–Jamiołkowski isomorphism establishes a one-to-one correspondence between quantum channels and quantum states, one can apply Theorem~\ref{thm:simultaneous-purification-states} to the Choi operators of a non-signalling instrument assemblage to obtain a simultaneous purification directly in the Choi picture.
More explicitly, 
consider the Choi operators of the instruments
\begin{equation}
    \Tilde{\mathcal{I}}_{a|x} 
= \sum_{i,j} \ketbra{i}{j} \otimes \mathcal{I}_{a|x} \left( \ketbra{i}{j} \right)
    \; \in \mathsf{B}(\mathcal{H}_{I}\otimes \mathcal{H}_{O})
\end{equation}
and consequently
\begin{equation}
    \Tilde{\mathcal{I}} = \sum_a \Tilde{\mathcal{I}}_{a|x}
= \sum_{i,j} \ketbra{i}{j} \otimes \mathcal{I}(\ketbra{i}{j})
    \; \in \mathsf{B}(\mathcal{H}_{I}\otimes \mathcal{H}_{O}).
\end{equation}
Then, inspired by the proof of Theorem~\ref{thm:simultaneous-purification-states}, we define the purification of $\Tilde{\mathcal{I}} \in \mathsf{B}(\mathcal{H}_{I}\otimes \mathcal{H}_{O})$ by
\begin{align}
    \ket{\psi}:=\sum_{ij} \ket{ij} \otimes \Tilde{\mathcal{I}}^{1/2} \ket{ij} \; \in \mathcal{H}_A \otimes \mathcal{H}_I \otimes \mathcal{H}_O,
\end{align}
with the auxiliary Hilbert space $\mathcal{H}_A \simeq \mathcal{H}_I \otimes \mathcal{H}_O$.
Note that since $\Tr_{AO} \ketbra{\psi}{\psi}=\mathds{1}_I$, $\ketbra{\psi}{\psi}$ is the Choi operator of an isometric quantum channel $\Lambda(\cdot)= V (\cdot) V^\dagger$ from $\mathsf{B}(\mathcal{H}_{I} )$ to $\mathsf{B}(\mathcal{H}_{A} \otimes \mathcal{H}_{O})$.
Now, we obtain the measurements defined on the auxiliary space $\mathcal{H}_{A}$ as
\begin{align}
M_{a|x} &\!=\! \left[ \Tilde{\mathcal{I}}^{-1/2} \Tilde{\mathcal{I}}_{a|x} \Tilde{\mathcal{I}}^{-1/2}\right]^T.
\end{align}
As before, a direct calculation shows that
\begin{equation}
    \Tilde{\mathcal{I}}_{a|x} = \Tr_{A} \left[ \left( M_{a|x} \otimes \mathds{1} \right) \ketbra{\psi}{\psi} \right],
\end{equation}
and hence for all $\rho \in \mathsf{B}(\mathcal{H}_I)$ also that
\begin{equation}
    \mathcal{I}_{a|x}(\rho)= \Tr_A \left[ \left(M_{a|x} \otimes \mathds{1} \right) \Lambda(\rho) \right],
\end{equation}
which concludes the proof.
\end{proof}

The Kraus decomposition proof is very explicit, however the Choi approach is best-suited to consider higher-order generalisations.
In the following, we prove that the Choi approach allows us to obtain simultaneous purification theorems not only for instruments, but also for super-instruments and arbitrary other quantum objects.

To formally capture the notion of general quantum objects, we adopt the definition proposed in~\cite{Milz_2024} which is formulated directly in the Choi picture.
Note that we modified the original definition to allow for sub-normalization instead of strict tracial normalization of quantum objects.
This allows us to extend the original definition to the sub-normalized elements of assemblages.
\begin{definition}\label{def:quantum-object}
    A set of linear operators $\mathcal{S}\subseteq \mathsf{B}(\mathcal{H})$ is called a quantum object set, if there exist a linear projective map $\mathcal{P}$ and a number $\gamma \in \mathbb{R}$ such that for every element $\Tilde{O} \in \mathsf{B}(\mathcal{H}$, it holds that $\Tilde{O} \in \mathcal{S}$ if and only if:
    \begin{enumerate}
        \item Positive semidefiniteness: $\Tilde{O} \geq 0$;
        \item Structural consistency: $\mathcal{P}(\Tilde{O})= \Tilde{O}$;
        \item Sub-normalization: $\Tr[\Tilde{O}] \leq \gamma$.
    \end{enumerate}
    Furthermore, a family $\{\tilde O_{a|x}\}_{a,x} \subseteq \mathcal{S}$ of quantum objects is called an assemblage if $\Tr[ \sum_a \tilde O_{a|x} ] = \gamma$ for all choices of $x$.
\end{definition}
Quantum states and their assemblages are the simplest example, and satisfy the definition for $\mathcal{P}$ being the identity map and $\gamma=1$.
Choi operators of quantum channels with input space $\mathcal{H}_I$ and output space $\mathcal{H}_O$ are captured by letting $\mathcal{H} = \mathcal{H}_I \otimes \mathcal{H}_O$, $\gamma = d_I$, and
\begin{equation}
    \mathcal{P}[\tilde O] := \tilde O - \Tr_O [\tilde O] \otimes \frac{\mathds{1}_O}{d_O} + \Tr_{IO} [\tilde O] \frac{\mathds{1}_{IO}}{d_I d_O},
\end{equation}
where $d_I = \dim (\mathcal{H}_I)$ and $d_O = \dim (\mathcal{H}_O)$.
It is not complicated to show that the set of Choi matrices of instruments and super-instruments also satisfy this condition; however this definition also captures the more general objects of Choi matrices of quantum combs, process matrices, multi-round process matrices, and arbitrary higher-order quantum maps.
The linear projection $\mathcal{P}$ enforces the specific causal or signalling constraints of the objects; for more details and concrete examples on how to choose $\mathcal{P}$ and $\gamma$ to represent specific higher-order structures, we refer to \cite{Milz_2024, simmons2022higherordercausaltheoriesmodels, hoffreumon2024projectivecharacterizationhigherorderquantum}.

We now have the tools to formulate the general simultaneous purification theorem for quantum objects.

\begin{figure*}[t]
    \begin{tikzpicture}[scale=0.3, every node/.append style={scale=0.7}]
\begin{scope}[shift={(0,0)}]
    
\fill[gray!30] (0,-4) -- (2,-4) -- (2,-1) -- (6,-1) -- (6,-4) -- (8,-4) -- (8,-1) -- (12,-1) -- (12,-4) -- (14,-4) -- (14,-1) -- (18,-1) -- (18,-4) -- (20,-4) -- (20,4) -- (15,4) -- (15,1) -- (11,1) -- (11,4) -- (9,4) -- (9,1) -- (5,1) -- (5,4) -- (0,4) -- cycle;
    
\draw[line width=0.7pt] (6,2) -- (8,2) -- (8,5) -- (12,5) -- (12,2) -- (14,2) -- (14,7) -- (6,7) -- cycle;
\draw[double, line width=0.7pt] (6,6) -- (5,6) node[left] {$x$};
    \draw[double, line width=0.7pt] (14,6) -- (15,6) node[right] {$a$};
\draw[line width=0.7pt] (5,3) -- (6,3);
    \draw[line width=0.7pt] (8,3) -- (9,3);
    \draw[line width=0.7pt] (11,3) -- (12,3);
    \draw[line width=0.7pt] (14,3) -- (15,3);
\node at (10,6) {\large$\mathcal{S}_{a|x}$};
    
\draw[line width=0.7pt] (3,-7) -- (17,-7) -- (17,-2) -- (15,-2) -- (15,-5) -- (11,-5) -- (11,-2) -- (9,-2) -- (9,-5) -- (5,-5) -- (5,-2) -- (3,-2) -- cycle;
\draw[double, line width=0.7pt] (3,-6) -- (2,-6) node[left] {$y$};
    \draw[double, line width=0.7pt] (17,-6) -- (18,-6) node[right] {$b$};
\draw[line width=0.7pt] (2,-3) -- (3,-3);
    \draw[line width=0.7pt] (5,-3) -- (6,-3);
    \draw[line width=0.7pt] (8,-3) -- (9,-3);
    \draw[line width=0.7pt] (11,-3) -- (12,-3);
    \draw[line width=0.7pt] (14,-3) -- (15,-3);
    \draw[line width=0.7pt] (17,-3) -- (18,-3);
\node at (10,-6) {\large$\mathcal{T}_{b|y}$};

    \end{scope} 

\node at (22,0) {\large$=$};

\begin{scope}[shift={(24,0)}]

\fill[col5!10, rounded corners] (-0.5,-7.5) -- (20.5,-7.5) -- (20.5,7.5) -- (-0.5,7.5) -- cycle;
    
\fill[gray!30] (0,-4) -- (2,-4) -- (2,-1) -- (6,-1) -- (6,-4) -- (8,-4) -- (8,-1) -- (12,-1) -- (12,-4) -- (14,-4) -- (14,-1) -- (18,-1) -- (18,-4) -- (20,-4) -- (20,4) -- (15,4) -- (15,1) -- (11,1) -- (11,4) -- (9,4) -- (9,1) -- (5,1) -- (5,4) -- (0,4) -- cycle;
    
\draw[line width=0.7pt] (6,2) -- (8,2) -- (8,5) -- (12,5) -- (12,2) -- (14,2) -- (14,7) -- (6,7) -- cycle;
\draw[line width=0.7pt] (14,6) -- (23,6);
\draw[line width=0.7pt] (5,3) -- (6,3);
    \draw[line width=0.7pt] (8,3) -- (9,3);
    \draw[line width=0.7pt] (11,3) -- (12,3);
    \draw[line width=0.7pt] (14,3) -- (15,3);
\node at (10,6) {\large$\Sigma$};
\draw[line width=0.7pt] (23,5) -- (26,5) -- (26,8) -- (23,8) -- cycle;
\draw[double, line width=0.7pt] (23,7) -- (22,7) node[left] {$x$};
    \draw[double, line width=0.7pt] (26,7) -- (27,7) node[right] {$a$};
\node at (24.5,6.5) {\large $M^A_{a|x}$};
    
\draw[line width=0.7pt] (3,-7) -- (17,-7) -- (17,-2) -- (15,-2) -- (15,-5) -- (11,-5) -- (11,-2) -- (9,-2) -- (9,-5) -- (5,-5) -- (5,-2) -- (3,-2) -- cycle;
\draw[line width=0.7pt] (17,-6) -- (23,-6);
\draw[line width=0.7pt] (2,-3) -- (3,-3);
    \draw[line width=0.7pt] (5,-3) -- (6,-3);
    \draw[line width=0.7pt] (8,-3) -- (9,-3);
    \draw[line width=0.7pt] (11,-3) -- (12,-3);
    \draw[line width=0.7pt] (14,-3) -- (15,-3);
    \draw[line width=0.7pt] (17,-3) -- (18,-3);
\node at (10,-6) {\large$T$};
\draw[line width=0.7pt] (23,-7) -- (26,-7) -- (26,-4) -- (23,-4) -- cycle;
\draw[double, line width=0.7pt] (23,-5) -- (22,-5) node[left] {$y$};
    \draw[double, line width=0.7pt] (26,-5) -- (27,-5) node[right] {$b$};
\node at (24.5,-5.5) {\large $M^B_{b|y}$};

    \end{scope}

\end{tikzpicture}     \caption{\justifying A bipartite scenario in which the network is modelled by a multi-round process matrix, and the players are described by assemblages of quantum combs. Time is flowing from left to right.}
    \label{fig:multi-round-process-matrix}
\end{figure*}
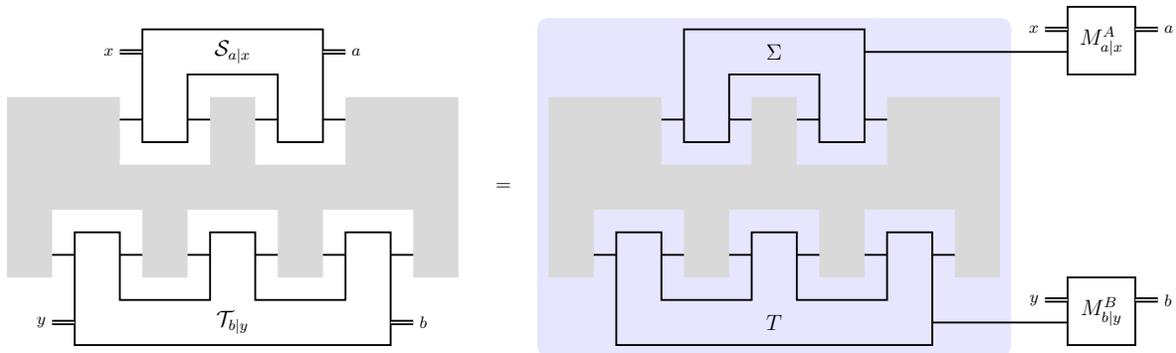

\begin{theorem}\label{thm:simultaneous-purification-quantum-object}
    Let $\mathcal{S} \subseteq \mathsf{B}(\mathcal{H})$ be a quantum object set characterized by $\mathcal{P}$ and $\gamma$, and $\{\tilde O_{a|x}\}_{a,x} \subseteq \mathcal{S}$ be an assemblage of quantum objects. Then, the following two conditions are equivalent:
    \begin{enumerate}
        \item The assemblage $\{\tilde O_{a|x}\}_{a,x}$ is no-signalling, \emph{i.e.}, there exists a quantum object $\tilde O \in \mathcal{S}$, such that
            \begin{equation}
                \sum_a \tilde O_{a|x} = \tilde O, \quad \text{for all } x. \label{eqn:non-signalling-assemblage-object}
            \end{equation}
        \item There exists a (finite-dimensional) auxiliary Hilbert space $\mathcal{H}_A$, a quantum object set $\mathcal{S}' \subseteq \mathsf{B}(\mathcal{H}_A \otimes \mathcal{H})$ characterized by $\mathcal{P}'$ and $\gamma$, where
        \begin{align}
            &\mathcal{P}'[\tilde W] := \tilde W + \mathds{1}_A \otimes \mathcal{P} ( \Tr_A [ \tilde W ] ) - \mathds{1}_A \otimes \Tr_A [ \tilde W ], \nonumber \\
            &\qquad \text{for } \tilde W \in \mathsf{B}(\mathcal{H}_A \otimes \mathcal{H}), \label{eqn:extended-projective-map}
        \end{align}
        a quantum object $\tilde \Omega \in \mathcal{S}'$, and a family of measurements $\{M_{a|x}\}_{a,x} \subseteq \mathsf{B}(\mathcal{H}_A)$, such that
            \begin{equation}
                \tilde O_{a|x} = \Tr_A \left[ \left( M_{a|x} \otimes \mathds{1} \right) \tilde \Omega \right]. \label{eqn:decomposition-assemblage-object}
            \end{equation}
        Furthermore, $\tilde \Omega$ can be chosen to be pure, \emph{i.e.}, of the form $\tilde \Omega = \ketbra{\omega}{\omega}$ for some $\ket{\omega} \in \mathcal{H}_A \otimes \mathcal{H}$.
    \end{enumerate}
\end{theorem}
\begin{proof}
    First observe that $\mathcal{P}'$ is a projective map, which follows from a direct calculation, using that $\mathcal{P}$ is projective.
    
    The implication $(2) \Rightarrow (1)$ follows immediately from the linearity of the trace and the completeness of the measurement.
    Assume (2) holds. Summing \eqref{eqn:decomposition-assemblage-object} over $a$, and using the completeness relation $\sum_a M_{a|x} = \mathds{1}_A$, we obtain:
    \begin{equation}
        \sum_a \tilde O_{a|x} = \Tr_A \left[ \left( \sum_a M_{a|x} \otimes \mathds{1} \right) \tilde \Omega \right] = \Tr_A \left[ \tilde \Omega \right],
    \end{equation}
    Let us define $\tilde O := \Tr_A [\tilde \Omega]$.
    As $\{\tilde O_{a|x}\}_{a,x}$ is an assemblage of quantum objects in $\mathcal{S}$, also $\tilde O = \sum_a \tilde O_{a|x} \in \mathcal{S}$, which concludes the proof of this direction.

    We now prove the implication $(1) \Rightarrow (2)$.
    Assume the assemblage $\{\tilde O_{a|x}\}_{a,x} \subseteq \mathcal{S}$ is no-signalling. By Definition~\ref{def:quantum-object}, the elements $\tilde O_{a|x}$ are positive semidefinite operators acting on $\mathcal{H}$.
    Since $\mathcal{S} \subseteq \mathsf{B}(\mathcal{H})$, we can treat the assemblage $\{\tilde O_{a|x}\}_{a,x}$ simply as a no-signalling (non-normalized) assemblage of states. We can therefore invoke Theorem~\ref{thm:simultaneous-purification-states} directly, which guarantees the existence of an auxiliary Hilbert space $\mathcal{H}_A$, an operator $\tilde \Omega = \ketbra{\omega}{\omega} \in \mathsf{B}(\mathcal{H}_A \otimes \mathcal{H})$ for some $\ket{\omega} \in \mathcal{H}_A \otimes \mathcal{H}$, and a POVM $\{M_{a|x}\}_{a,x} \subseteq \mathsf{B}(\mathcal{H}_A)$ such that:
    \begin{equation}
        \tilde O_{a|x} = \Tr_A \left[ \left( M_{a|x} \otimes \mathds{1} \right) \tilde \Omega \right].
    \end{equation}
    Consistent with the proof of Theorem~\ref{thm:simultaneous-purification-states}, one constructive way of choosing $\ket{\omega}$ and $M_{a|x}$ is given by
    \begin{align}
        &\ket{\omega} = \sum_i \ket{i} \otimes \tilde{O}^{1/2} \ket{i}, \\
        &M_{a|x} = \left[ \tilde{O}^{-1/2} \tilde{O}_{a|x} \tilde{O}^{-1/2} \right]^T,
    \end{align}
    where $\{\ket{i}\}_{i}$ constitutes a basis of the auxiliary Hilbert space $\mathcal{H}_A \simeq \mathcal{H}$.
    It remains to show that the purification $\tilde \Omega$ is a valid quantum object of the extended type.
    \begin{enumerate}
        \item Positivity: By Theorem~\ref{thm:simultaneous-purification-states}, $\tilde \Omega$ is a pure state (up to normalization), hence $\tilde \Omega \geq 0$.
        \item Normalization: Summing the decomposition yields $\Tr_A[\tilde \Omega] = \tilde O$. Thus, $\Tr[\tilde \Omega] = \Tr[\tilde O] \leq \gamma$, satisfying the normalization condition.
        \item Structural consistency: The type of the extended object $\tilde \Omega$ is defined by the constraints on the system $\mathcal{H}$, with $\mathcal{H}_A$ acting as an unconstrained auxiliary output. Formally, if $\mathcal{P}$ is the projector for $\mathcal{S}$, the projector for the extended type $\mathcal{S}'$ is $\mathcal{P}'$ as in \eqref{eqn:extended-projective-map}.
        Since $\Tr_A[\tilde \Omega] = \tilde O \in \mathcal{S}$, it holds that $\mathcal{P}(\Tr_A[\tilde \Omega]) = \Tr_A[\tilde \Omega]$, and hence
        \begin{equation}
            \mathcal{P}' ( \tilde \Omega ) = \tilde \Omega + \mathds{1}_A \otimes \mathcal{P} ( \Tr_A [ \tilde \Omega ] ) - \mathds{1}_A \otimes \Tr_A [ \tilde \Omega ] = \tilde \Omega.
        \end{equation}
    \end{enumerate}
    Thus, $\tilde \Omega \in \mathcal{S}'$ is a valid pure quantum object of the required type, completing the proof.
\end{proof}
In the case that the quantum object set $\mathcal{S}$ was describing quantum channels, super-channels, or quantum combs, the extended type $\mathcal{S}'$ from Theorem~\ref{thm:simultaneous-purification-quantum-object} corresponds to adding an additional, unconstrained output state in the auxiliary Hilbert space $\mathcal{H}_A$ to the original type $\mathcal{S}$.

This framework, that systematically allows us to map complicated quantum process to the well-understood case of states, has the advantage 
of providing concrete realizations for the purified objects and auxiliary Hilbert spaces, which were unknown before, even for the simple cases of instruments and super-instruments.

\section{Further communication scenarios}\label{app:example-comm-scenarios}
Figure~\ref{fig:multi-round-process-matrix} depicts non-signalling decompositions and correlations arising in a bipartite communication scenario in which the two players are allowed to send messages in both directions, and can make use of internal quantum memory.
The players are modelled by no-signalling assemblages of quantum combs, whose number of slots corresponds to the number of exchanged messages.
The network in this scenario, \emph{i.e.}, the way in which the two players exchange messages, is allowed to be of indefinite causal order, and is therefore mathematically captured by a multi-round process matrix~\cite{Hoffreumon_2021}.
All correlations arising in such scenario are guaranteed to admit a quantum Bell nonlocal realization, as follows immediately from the decomposition proved in Appendix~\ref{app:SGHJW-proofs} and depicted in Fig.~\ref{fig:multi-round-process-matrix}.
By the arguments from Appendix~\ref{app:SGHJW-proofs}, this scenario can be further generalized to the -- admittedly more difficult to represent graphically -- situations of more than two participating players and scenarios in which the messages exchanged by any player follow no definite causal order, requiring us to model the players as non-signalling assemblages of multi-round process matrices as well.

\end{document}